%% file: SIM_Universal_v4_Submitted.tex
 \documentclass[conference]{IEEEtran}
\IEEEoverridecommandlockouts
\usepackage [latin1]{inputenc}

\usepackage{graphicx,cite,acronym,amsmath,amssymb,amsfonts,color,subfigure,floatflt,stfloats,bm,textgreek,multirow,amsthm}
\usepackage{epsfig,epstopdf,psfrag,booktabs}

\usepackage[table]{xcolor}
\usepackage{bm}

\usepackage[cmintegrals]{newtxmath}
\usepackage{graphics,hhline,array,cite,bbm}

\usepackage{latexsym}
\usepackage{times}
\usepackage{makeidx}

\usepackage{mathtools}
\usepackage{graphicx}
\usepackage{hyperref}
\usepackage{xcolor}
\usepackage{enumitem}
\usepackage[normalem]{ulem}

\DeclareGraphicsExtensions{.png,.jpg,.eps}

\input{acronyms}

\include{def_EM_SP}

\theoremstyle{remark}
\newtheorem{remark}{Remark} 

\begin{document}

\title{On the Universal Approximation Capability of Stacked Intelligent Metasurfaces}

\author{\IEEEauthorblockN{
Davide Dardari
}\\
\IEEEauthorblockA{DEI, Universit\`a di Bologna, 40136 Bologna, Italy, Email: davide.dardari@unibo.it}
\IEEEauthorblockA{National Laboratory of Wireless Communications (WiLab), CNIT, 40136 Bologna, Italy}\\
}


\maketitle

\begin{abstract}

Wave-domain processing is an emerging paradigm in which signal-processing operations are partially shifted from the digital domain to the electromagnetic (EM) domain to reduce complexity, energy consumption, and latency in next-generation wireless systems. Stacked intelligent metasurfaces (SIMs) are a promising technology for wave manipulation. Although numerous studies have addressed the numerical design of SIMs for specific tasks, no theoretical results have yet established whether a prescribed SIM architecture can approximate an arbitrary complex linear mapping, i.e., whether it possesses a universal approximation capability. To address this question, this paper establishes fundamental theorems that provide necessary and sufficient closed-form conditions under which a given SIM architecture can serve as a universal approximator of complex linear transformations.
 
\end{abstract}

\begin{IEEEkeywords}
Wave-domain processing, stacked intelligent metasurface, universal approximation, holographic MIMO.
\end{IEEEkeywords}

\section{Introduction}

It is widely recognized that satisfying the stringent communication and sensing requirements of next-generation wireless networks will require highly dense and/or large-scale massive \ac{MIMO} antenna systems \cite{BjoChaHeaMarMezSanRusCasJunDem:24,Pre:J24,YouCaiLiuDiRDumYen:25}. However, the deployment of massive antenna arrays entails significant challenges in terms of hardware complexity, processing latency, and energy consumption, thereby raising concerns about the sustainability of future wireless infrastructures.

A promising approach to mitigating these challenges is to distribute signal-processing functions across the digital, analog, and wave domains, leading to the so-called \emph{tri-hybrid} architecture \cite{CasYanChaHea:26}. In this framework, wave-domain processing performs part of the signal-processing operations directly at the \ac{EM} level, in accordance with the \ac{ESIT} paradigm \cite{DarTorPasDec:J26, ZhuWanDaiDebPoo:24}, enabled by recent advances in reconfigurable antenna technologies and metamaterials.

Among the wave-based processing technologies proposed to date, \acp{SIM} have attracted particular attention \cite{AnXuetAl:J23}.
%
%
A \ac{SIM} consists of multiple programmable layers separated by free-space or guided-wave propagation regions. Cascaded wave-matter interactions provide a multilayer analog processor in which the propagation operators are fixed by geometry and frequency, while the diagonal layer responses are programmable (typically through varactors, liquid crystals, or PIN diodes). 
The resulting multilayer cascade can realize substantially richer wave-domain transformations than a single layer and has motivated applications to analog beamforming, \ac{MIMO} precoding and combining, wideband processing, \ac{ISAC}, holographic \ac{MIMO}, and electromagnetic-domain computing~\cite{AnXuetAl:J23,Li:26,HasAnDiRDebYue:24,AnYueGuaDiRDebPooHan:24,AnXuetAl5:25,FabTorDar:C25}. 
A recent survey by Sheemar et al. consolidates modeling approaches (cascaded operators, multiport impedance descriptions, network representations), and applications to front-end MIMO processing, near-field and wideband transmission, learning-based control, and \ac{ISAC}~\cite{She:2026}.  

A central theoretical question, distinct from, and prior to, the question of {\em how} to tune a \ac{SIM} for a given task, is whether a prescribed \ac{SIM} architecture can implement, exactly or approximately, an {\em arbitrary} complex linear mapping between an $S$-dimensional input space (e.g., $S$ RF feeds or data streams) and an $M$-dimensional target space (e.g., the outputs of the last layer of a transmitting \ac{SIM}), that is, act as a \emph{linear universal approximator}. Existing works almost invariably optimize the meta-atom coefficients for one specified task, such as beamforming,  given a channel estimate, and \ac{MIMO} precoding, by using different methods, for instance, alternating optimization, projected gradient descent, or manifold methods~\cite{She:2026}. Such algorithms are indispensable in practice, but successful convergence for a handful of selected tasks does not, by itself, establish that the underlying architecture has sufficient intrinsic degrees of freedom to approximate the {\em entire} space of transfer matrices $\Complex^{M\times S}$: a poorly designed cascade (too few layers, too few elements, or a geometry that makes successive layers act as a single effective layer) may simply be incapable of realizing certain tasks, no matter how the optimizer is tuned, initialized, or how long it is run.

This paper addresses this question for the first time, to the best of the author's knowledge, through the derivation of the necessary conditions and sufficient conditions for a \ac{SIM} to act as a linear universal approximator, making use of algebraic topology arguments. For a given \ac{SIM} architecture in terms of propagation matrices, such closed-form conditions are algorithm-agnostic and can be checked once, off-line, from the propagation matrices alone, to certify the universality of the architecture and then license (or not license) any subsequent task-specific synthesis.

\begin{figure}[!t]
\centering\includegraphics[width=1\columnwidth]{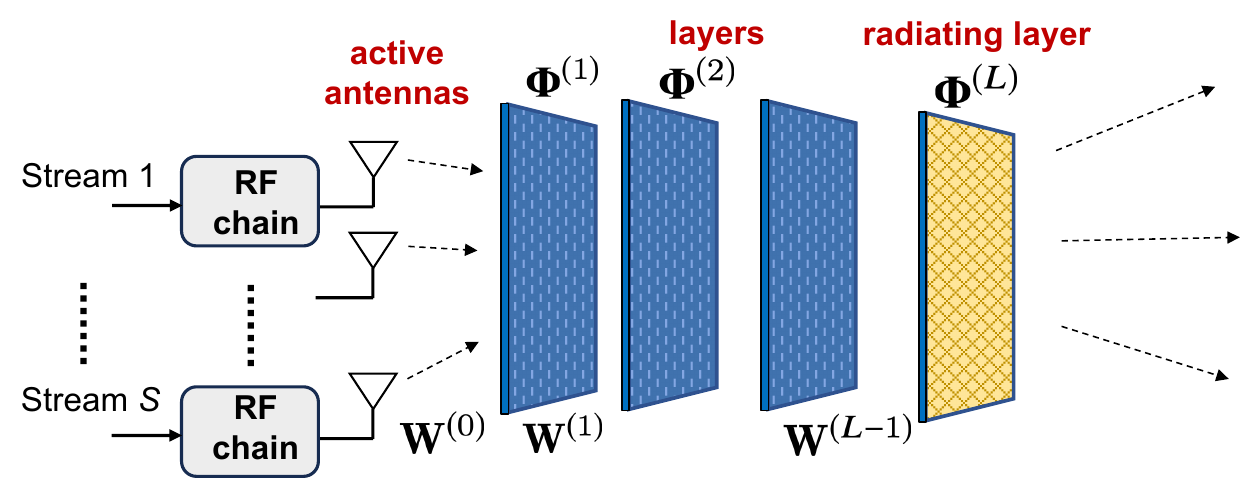}
\caption{Transmitting SIM architecture.} 
\label{Fig:SIM}
\end{figure}

\section{SIM Model and Problem Statement}
\label{Sec:Model}

\subsection{SIM model}
Consider a transmitting \ac{SIM} with $L$ programmable layers of $M$ reconfigurable meta-atoms each, preceded by $S$ RF feeds (data streams), as shown in Fig. \ref{Fig:SIM}.

Denote with $\boldW_{\ell} \in \Complex^{M \times M}$, for $\ell=1,2, \ldots , L-1$, the propagation matrix between layer $\ell$ and layer $\ell+1$, and with $\boldW_0 \in \Complex^{M \times S}$, the propagation matrix between the $S$ RF feeds and the first layer.  The fixed propagation matrices are determined by the layer geometry, technology, and the operating frequency.
The only online design freedom is given by the reconfigurable meta-atoms at each layer. Specifically, the transfer matrix of layer $\ell$ is denoted to as $\Phi_{\ell}(\Parametersl)$, where $\Parametersl=\{\Parameterlm \}$, for $m=1,2, \ldots, M$, are the $M$ tunable complex parameters.  Assuming no coupling between meta-atoms, it is $\Phi_{\ell}(\Parametersl)=\diag{\Parametersl}$. For a passive \ac{SIM}, the following constraint must hold: $|\Parameterlm|\le 1$, $\forall \ell, m$. We gather all the reconfigurable parameters in the vector $\Parameters=\{ \Parametersl\}$ whose size is $N=M\, L$.
The resulting transfer matrix of the \ac{SIM} is~\cite{AnXuetAl:J23} 
\begin{equation}
\boldH(\Parameters)=\Phi_{L}(\Parameters_{L}) \boldW_{L-1} \Phi_{L-1}(\Parameters_{L-1}) \ldots \boldW_1 \Phi_{1}(\Parameters_{1}) \boldW_0  \in \Complex^{M \times S}\, .  
\label{eq:H}
\end{equation}

Define the normalized transfer matrix 
\begin{equation}
\label{eq:tH}
\widetilde{\boldH}(\Parameters)= \frac{\boldH(\Parameters)}{\| \boldH(\Parameters)\|_{\text{F}}}
\end{equation}
where $\left \|  \cdot  \right \|_{\text{F}}$ denotes the the Frobenius norm. In a typical design problem, one is interested in obtaining a target response apart from a constant that can be absorbed in the link budget. 
In the following, without loss of generality, we allow the parameters to take any complex value, i.e., $\Parameters \in \Complex^{N}$. In fact, thanks to the specific structure of \eqref{eq:H}, any nonzero complex configuration can be rescaled so that $|\Parameterlm|\le 1$ for all entries without altering \eqref{eq:tH}. 

\subsection{Problem statement}

For further convenience, define the vectorized version $\boldh(\Parameters)=\mathrm{vec}(\boldH(\Parameters))$ of $\boldH(\Parameters)$, where the operator $\mathrm{vec}(\boldH)$ stacks the columns of matrix $\boldH$ into a column vector. The resulting parametric vector $\boldh(\Parameters)$ can be seen as a mapping $\boldh: \Complex^N \to \Complex^{K}$ between the complex parameter space of dimension $N$ and the complex output space of dimension $K=MS$ associated with all transformation matrices $\boldH(\Parameters)$ of size $M\times S$.

\begin{definition}[Linear universal approximation]
The \ac{SIM} acts as a linear universal approximator if, for any target mapping $\boldh_0$, there exists a parameter configuration $\Parameters^*$ such that $\| \boldh(\Parameters^*) - \boldh_0\| < \epsilon$ for any desired accuracy level $\epsilon$. 
\end{definition}

In terms of algebraic geometry, the universal approximation objective corresponds to the property
\begin{equation}
\label{eq:UA}
\overline{\boldh \left (\Complex^N \right )}^{\,\mathrm{Eucl}}=\Complex^{K}  
\end{equation}
where $\overline{\Ucal}^{\,\mathrm{Eucl}}$ denotes the closure of the set $\Ucal$ according to the Euclidean geometry.
The meaning of \eqref{eq:UA} is that the image $\boldh\left (\Complex^N\right )$ of the mapping $\boldh$ is the entire output space apart from a subset of Lebesgue measure zero. In other words, $\boldh$ is a universal approximator if and only if it is \emph{surjective almost everywhere}. By the definition of the closure of a set, for any point in the output space not in the image of $\boldh$, we can find a sequence of points in the image converging to it, i.e., with any desired accuracy.

Therefore, the problem addressed in this paper is to establish practical necessary and sufficient conditions under which a fixed \ac{SIM} architecture, hence, fixed $\boldW_0, \boldW_1, \ldots, \boldW_{L-1}$, acts as a linear universal approximator, i.e., satisfies \eqref{eq:UA}.

\section{Necessary and Sufficient Conditions}

Before enunciating the necessary and sufficient conditions for \eqref{eq:UA}, some introductory definitions and properties are needed. 

\begin{proposition}
The map $\boldh$ is polynomial and holomorphic.  \end{proposition}
\begin{proof}
Each $\Phi_{\ell}(\Parametersl)$ depends linearly on its diagonal entries. Eqn. \eqref{eq:H} is a product of $2L-1$ matrix factors alternating fixed matrices and variable diagonal matrices. Since matrix multiplication is bilinear, the product is a function of monomials of total degree $L$ in the variables $\Parameterlm$. Being a finite sum of monomials, $\boldh$ is polynomial in $\Parameters$ and does not depend on its conjugate, hence it is holomorphic, i.e., complex analytic. 
\end{proof}

\begin{definition} [Jacobian of $\boldh$]
The Jacobian of $\boldh$ is denoted with $\boldJ_{\boldh}(\Parameters) \in \Complex^{K \times N}$, whose generic element is defined as $\left [ \boldJ_{\boldh}(\Parameters) \right ]_{k,n}=\frac{\partial \left [\boldh(\Parameters) \right ]_k}{\partial \Parameters_n}$, for $k=1,2, \ldots, K$, and $n=1,2, \ldots ,N$.
\end{definition}

\begin{theorem}[Necessary conditions for universal approximation]
\label{thm:Necessary}
Necessary but not sufficient conditions for $\boldH$, and hence $\boldh$, to be surjective almost everywhere are
\begin{equation}
\label{eq:Necessary}
\begin{array}{l}
\text{(a) }L \ge S \\
\text{(b) }\rank{\boldW_0}=S \, .
\end{array}
\end{equation}
\end{theorem}
\begin{proof}
Theorem \ref{thm:Necessary} can be easily proved starting from dimension arguments. In fact, surjectivity requires that the dimension of the domain is not smaller than the dimension of the output space. In our case, the dimension of the output space is $K=MS$, whereas the dimension of the parameter space is $N=ML$; then condition (a) follows. Regarding condition (b), if $\boldW_0$ were not full rank, then it would be impossible for $\boldH$ to be full rank and hence approximate any full rank target matrix.   
\end{proof}

Condition $L\ge S$ provides a lower bound on the number of layers required for universal approximation. Interestingly, such a bound depends only on the number of RF feeders, $S$, and not on the number of meta-atoms, $M$. However, without additional assumptions, this condition is not sufficient, since the propagation matrices may cause the image of $\boldh$ to lie in a non-dense subspace of $\Complex^{K}$.

\begin{theorem}[Sufficient condition for universal approximation]
\label{thm:Sufficient}
If there exists $\Parameters_0 \in \Complex^N$ such that
\begin{equation}
\label{eq:Sufficient}
\rank{\boldJ_{\boldh}(\Parameters_0)}=K
\end{equation}
then 
\begin{equation}
\label{eq:UA1}
\overline{\boldh \left (\Complex^N \right )}^{\,\mathrm{Eucl}}=\Complex^{K} \, .
\end{equation}  
Consequently, $\boldh$, and hence $\boldH$, are surjective almost everywhere, i.e., they are linear universal approximators.
\end{theorem}
\begin{proof}
See the Appendix.
\end{proof}

\begin{corollary}
\label{cor:suff}
A sufficient  but not necessary condition for $\boldH$, and hence $\boldh$, to be surjective almost everywhere is
\begin{equation}
\label{eq:J1}
\rank{\boldJ_1}=K
\end{equation}
where
\begin{align}
\left [\boldJ_1 \right ]_{:,n}&=\left ( \boldB_{\ell}^\top \bolde_m \right ) \otimes  \left ( \boldA_{\ell} \bolde_m \right )   \label{eq:defcor} \\
\boldA_{\ell}&=\boldW_{L-1} \boldW_{L-2} \ldots \boldW_{\ell}   \quad \boldA_{L}=\boldI_M \nonumber \\
\boldB_{\ell}&=\boldW_{\ell-1} \boldW_{\ell-2} \ldots \boldW_{1} \boldW_0  \nonumber \\
\bolde_m&=[0,0,1,0,\dots,0]^\top \quad \text{the one in the $m$th position} 
\nonumber 
\end{align}
for $n=(\ell-1)M +m$, $\ell=1,2, \ldots , L$, and $m=1,2, \ldots, M$, where $\otimes$ represents the Kronecker product.
\end{corollary}
\begin{proof}
This corollary is a particular case of Theorem \ref{thm:Sufficient} by selecting the particular configuration $\Parameters_0=\boldsymbol{1}_N=[1,1, \ldots , 1]$. With this choice, the diagonal matrices $\Phi_{\ell}$ become the identity matrix, i.e., $\Phi_{\ell}(\boldsymbol{1}_N)=\boldI_M, \forall \ell$. Through straightforward matrix manipulations, not reported here due to space constraints, the Jacobian of $\boldh(\Parameters)$ evaluated at $\boldsymbol{1}_N$, $\boldJ_1=\boldJ_{\boldh}(\boldsymbol{1}_N)$, is given by \eqref{eq:defcor}. Therefore, if \eqref{eq:J1} is satisfied, then we have found a full-rank configuration $\Parameters_0$ and the result \eqref{eq:UA1} of Theorem \ref{thm:Sufficient} follows.
\end{proof}

\begin{remark}
\item In the common case of identical $\boldW_{\ell}=\boldW$ matrices, for $\ell>0$, definitions \eqref{eq:defcor} simplify into $\boldA_{\ell}=\boldW^{L-\ell}$ and $\boldB_{\ell}=\boldW^{\ell-1} \boldW_0$
\end{remark}

\begin{remark}
\item For a given \ac{SIM}, the sufficient condition in Corollary \ref{cor:suff} is directly checkable in closed form through a simple \ac{SVD} operation.
\end{remark}

\begin{remark}
As discussed in the Appendix, Theorem \ref{thm:Sufficient} is not restricted to \acp{SIM}; rather, it applies to any mapping that is polynomial or rational over its domain
\end{remark}


\section{Numerical Examples}
\label{Sec:NumericalResults}

Although a numerical verification of the theorems established in the previous section is neither feasible nor necessary after their rigorous proof, the numerical examples presented in this section provide insight into the dependence of the performance on the \ac{SIM} architecture.

\begin{figure}[!t]
\centering\includegraphics[width=1\columnwidth]{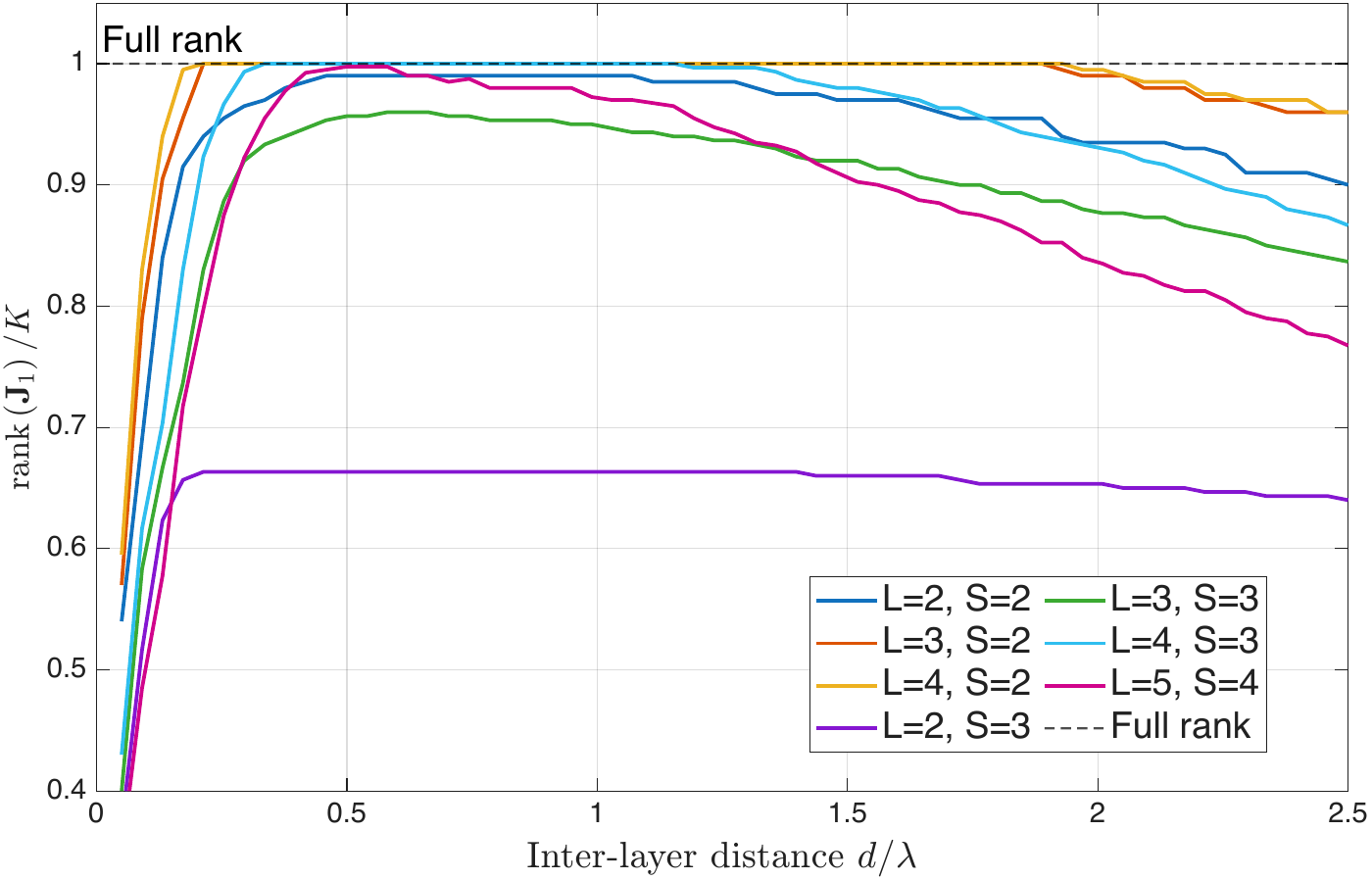}
\caption{Normalized numerical rank of the Jacobian  $\boldJ_1$ in Corollary \ref{cor:suff} versus the normalized inter-layer distance  $d/\lambda$ for different SIM depths $L$ and numbers of input streams $S$.} 
\label{Fig:Rank}
\end{figure}

We first analyze when condition \eqref{eq:J1} is satisfied for different interlayer distances $d$ and different numbers of layers $L$ and data streams $S$. 
In particular, each layer of the \ac{SIM} consists of a square panel of $M=10 \times 10$ meta-atoms spaced by $\lambda/2$, where $\lambda$ is the wavelength.
We model each layer-to-layer link with the scalar Rayleigh--Sommerfeld diffraction formula standard in the \ac{SIM} literature~\cite{AnXuetAl:J23}. Specifically, for a source point at position $\mathbf p$ on one layer and a destination point at $\mathbf q$ on the next, element area $A_t$, and $r=|\mathbf q-\mathbf p|$, it is
\begin{equation}
w(\mathbf p,\mathbf q)=\frac{A_t\cos\chi}{r}\left(\frac{1}{2\pi r}-\frac{j}{\lambda}\right)e^{j2\pi r/\lambda}
\label{eq:rs}
\end{equation}
where $\cos\chi$ is the obliquity factor (the cosine of the angle between the layer normal and the line joining $\mathbf p$ and $\mathbf q$, and $A_t=\lambda^2$. Recently, more refined models have been presented based on multi-port circuit theory \cite{AbrBarToc:25}. 

Fig.~\ref{Fig:Rank} reports the normalized numerical (effective) rank of
$\mathbf{J}_{\mathbf{1}}$, computed by setting an upper bound on the condition number set equal to  $10^{3}$, 
as a function of the normalized inter-layer distance $d/\lambda$, for different pairs $(L,S)$. As predicted by Theorem~\ref{thm:Necessary}, the configuration with $L<S$ cannot achieve
$\operatorname{rank}(\mathbf{J}_{\mathbf{1}})=K$, since the number of
complex programmable parameters, equal to $N=ML$, is smaller than the
dimension $K=MS$ of the complex output space. Conversely, for the
architectures satisfying $L\geq S$, the Jacobian can attain full rank over a broad range of inter-layer distances. Whenever
$\operatorname{rank}(\mathbf{J}_{\mathbf{1}})=K$, the readily checkable
sufficient condition in Corollary~III.5 is fulfilled, thereby certifying
that the corresponding SIM architecture is a linear universal
approximator.

The dependence on $d/\lambda$ further demonstrates that the layer-count
condition $L\geq S$, while necessary, is not sufficient by itself. In
particular, for very small inter-layer distances, the propagation between
adjacent layers provides limited spatial mixing, and the transformations
induced by different layers become nearly redundant, thus causing a rank drop.

As the inter-layer distance increases, diffraction mixes the fields across
the meta-atoms and makes the actions of successive programmable layers
more distinct, enabling the SIM to realize a richer class of linear transformations than a single-layer surface. This increased diversity can make the columns of $\mathbf{J}_{\mathbf{1}}$ linearly independent and allow the full-rank
condition to be satisfied.
Large inter-layer separations generally reduce the
magnitude of the propagation coefficients and can decrease the rank as well as the overall
power-transfer efficiency. This latter effect concerns the attenuation of the SIM cascade and is distinct from the algebraic universal-approximation property characterized by the rank of $\mathbf{J}_{\mathbf{1}}$.
From Fig. \ref{Fig:Rank}, it can be observed that for the considered model, the optimum value of the inter-layer distance is $d=\lambda/2$, which will be used subsequently. 

\begin{figure}[!t]
\centering\includegraphics[width=1\columnwidth]{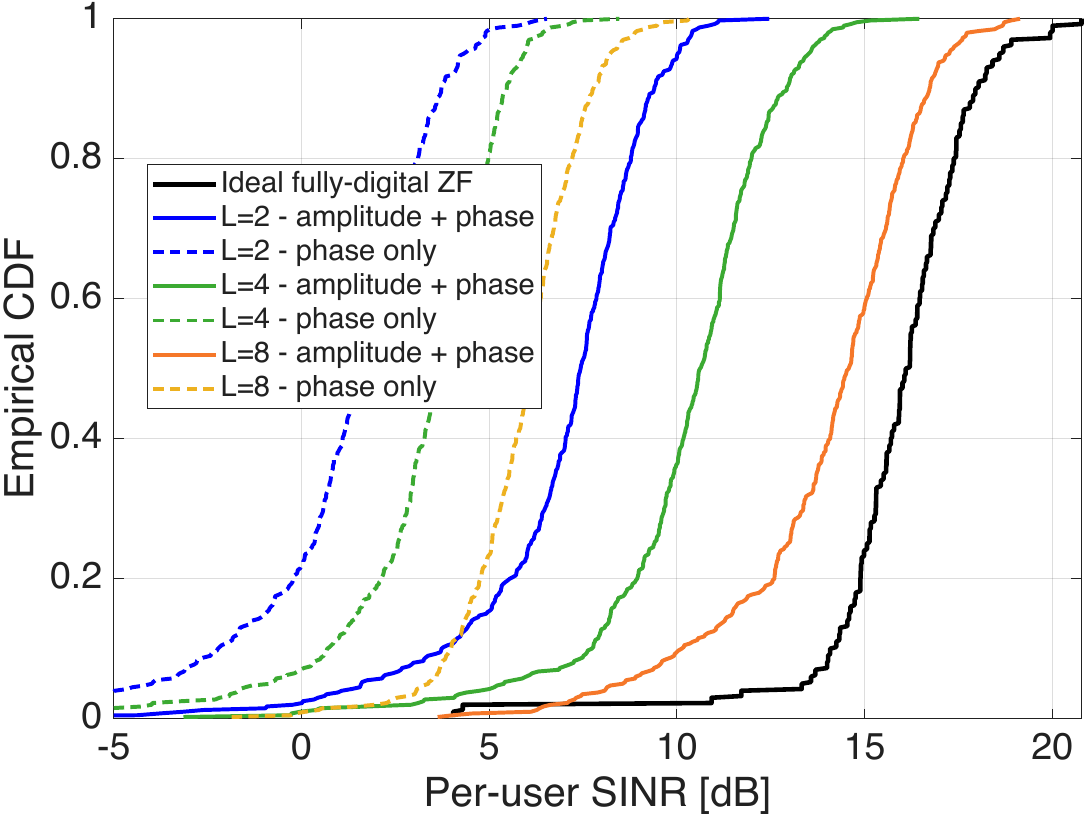}
\caption{Empirical CDF of the per-user SINR for different numbers of layers. Comparison between ideal full digital, amplitude-phase optimization, and only phase optimization.} 
\label{Fig:CDF}
\end{figure}

As a second example, we consider a transmitting \ac{BS} hosting a \ac{SIM} designed to realize a \ac{ZF} precoder in the wave domain to serve $S=4$ single-antenna users randomly deployed in a square area in front of the \ac{BS} between $10\,$m and $100\,$m. The carrier frequency is $f_c=3.5\,$GHz, the noise power is $\sigma^2=-90\,$dBm, and the transmit power is $\Ptx=-15\,$dBm.
For each scenario (100 Monte Carlo iterations), the parameters of the \ac{SIM} $\Parameters$ are computed by numerically optimizing the following constrained minimization problem~\cite{AnXuetAl:J23} 
\begin{align}
 \hat{\Parameters}&=\arg \min_{\Parameters} \left \|  \widetilde{\boldH}(\Parameters) - \widetilde{\boldH}_0 \right \| \nonumber  \\
 & \text{s.t. } |\Parameterlm| \le 1, \quad \text{for } \ell=1,2, \ldots L, \, m=1,2, \ldots, M 
\end{align}
where $\widetilde{\boldH}_0$ is the ideal normalized \ac{ZF} matrix computed by assuming perfect channel estimation. 
The \ac{SIM} coefficients are optimized through a multi-start projected gradient method with Armijo backtracking line search, with complex derivatives evaluated via Wirtinger calculus and feasibility enforced by projection onto the meta-atom constraint set \cite{Ber:B99}. 

Fig.~\ref{Fig:CDF} shows the per-user \ac{SINR} empirical \ac{CDF} when the fully digital and \ac{SIM} precoders are normalized to the same output-layer radiated power. Hence, the figure removes the attenuation of the \ac{SIM} and isolates the quality of the synthesized precoding transformation. The investigation of the joint optimization of \ac{SIM} power efficiency and the quality of the synthesized transformation will be the subject of future work. 
The fully digital \ac{ZF} benchmark ideally diagonalizes the multiuser channel, whereas the \ac{SIM} realizes an approximation of the desired mapping through a cascade of diagonal programmable layers separated by fixed propagation operators. Consequently, the \ac{SINR} gap is caused by residual mismatch in the effective channel, which reduces the desired gain and, more importantly, produces residual inter-user interference. For comparison, the case where only the optimization of the phase of $\Parameterlm$ is also reported, and the amplitude is kept constant at one. As can be noticed, the amplitude-and-phase \ac{SIM} outperforms the phase-only implementation because independent amplitude control
enlarges the feasible set and improves the approximation of the target \ac{ZF} matrix. 
As it can be noticed, as $L$ becomes larger than $S=4$, the performance improves significantly, getting close to that of the ideal full digital reference. The remaining gap can be ascribed to the non-optimality of the numerical minimum search that could fall into a local minimum.   

%
%
%
%
%

%

\section{Conclusion}
\label{Sec:Conclusion}


This paper establishes fundamental necessary and sufficient conditions for a \ac{SIM} architecture to universally approximate arbitrary complex linear maps. The resulting conditions are architecture-specific but independent of the particular design or optimization algorithm; for a given \ac{SIM} configuration, they provide a direct certificate of its universal approximation capability. 
More broadly, the theoretical framework developed in this paper is not limited to \acp{SIM} and can be applied to other classes of polynomial or rational mappings. While the present work focuses exclusively on the existence of the universal approximation property, future work will investigate efficiency-related aspects.


%
%
%
%
%
%

\section*{Acknowledgment}
The author would like to thank Nicol\'o Decarli, Mattia Fabiani, Francesco Miccoli, Malte Schellmann, Jose Mauricio Perdomo Zelaya, and Chan Zhou for the useful discussions. 
This work was supported by Huawei under the project EM-SP.

\section*{Appendix: Proof of Theorem \ref{thm:Sufficient}}

The proof of Theorem \ref{thm:Sufficient} makes use of results from algebraic geometry. 
We first recall the relevant algebraic-geometric definitions and properties related to the \emph{Zariski topology} \cite{Hartshorne:B77}.

\subsection{Zariski topology}
The Zariski topology is a topology that is well-suited for the study of polynomial equations in algebraic geometry and is defined on geometric objects called \emph{varieties}. An algebraic variety is defined as the set of solutions of a system of polynomial equations over the real or complex numbers. 

 \begin{definition}[Closed sets] In the Zariski topology, given a collection of polynomials, a closed set is the set of all points where all the polynomials evaluate to zero. 
 \end{definition}


\begin{definition}[Open sets] A Zariski open set is the complement of a Zariski closed set. 
\end{definition}
    
\begin{definition}[Zariski closure]
For \(\Ucal \subseteq\Complex^M\), its \emph{Zariski closure},
denoted by $\overline{\Ucal}^{\,\mathrm{Zar}}$, 
is the smallest Zariski-closed set containing \(\Ucal\). 
\end{definition}

\begin{definition}[Dominant map]
A polynomial map $\boldh:\mathbb{C}^N\to\Complex^K$ is called \emph{dominant} (\emph{Zariski dense}) if
$$
\overline{\boldh \left. (\Complex^N \right )}^{\,\mathrm{Zar}}
=\Complex^K.
$$
Equivalently, \(\boldh\) is dominant if the only polynomial $p$
satisfying $p(\boldh(\Parameters))=0, 
\, \forall \Parameters\in\Complex^N$
is the zero polynomial. 
\end{definition}

\begin{definition}[Constructible subset]
A subset of an algebraic variety is called
\emph{constructible} if it is a finite union of locally closed sets,
i.e., sets of the form $ \Ucal\cap \Zcal$,   where \(\Ucal\) is Zariski open and \(\Zcal\) is Zariski closed.  We also use the following standard consequence.
\end{definition}

 The Zariski topology is coarser than the Euclidean topology. 
 In fact, closed sets in the Zariski topology are \emph{rigid} because they are defined by polynomials: for instance, in the real line $\mathbb{R}$, the only Zariski closed sets are finite collections of points (because a non-zero polynomial can only have a finite number of roots). Instead, in a Euclidean space, closed sets are numerous and can take any shape, e.g., the interval $[0,1]$. Every Zariski closed set is also closed in the Euclidean topology. The converse is not true (e.g., a disk or a sphere is Euclidean closed but not Zariski closed). Any two non-empty open sets must intersect and cannot be disjoint. In fact, in the Euclidean topology, one can always take two distinct points and isolate them within two disjoint open balls. In the Zariski topology, this is impossible: open sets are so large that any two non-empty open sets will always overlap.

 
 \begin{figure}[!t]
\centering\includegraphics[width=1\columnwidth]{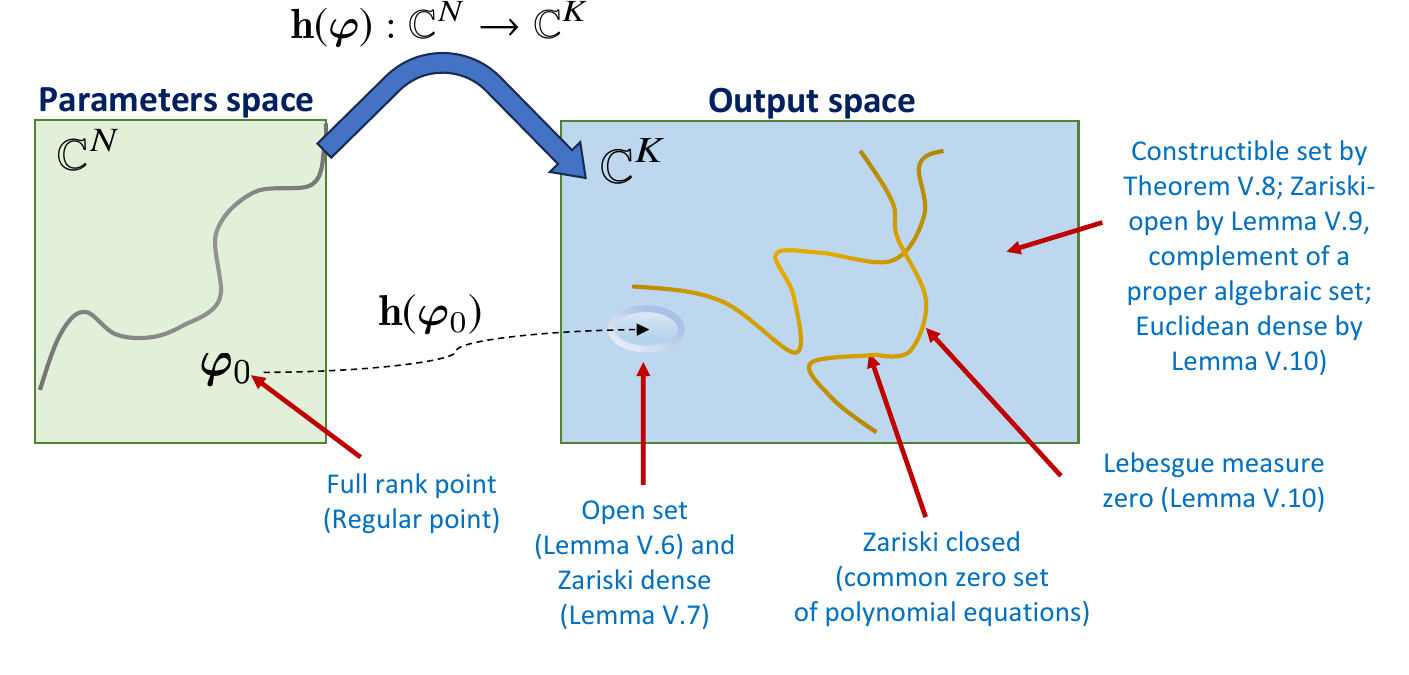}
\caption{Sketch of the proof of Theorem \ref{thm:Sufficient}.} 
\label{Fig:Proof}
\end{figure}

The proof of Theorem \ref{thm:Sufficient} consists of the following steps that are sketched in Fig. \ref{Fig:Proof}.

\subsection{Step 1: Full rank implies that the image contains an open set}


\begin{lemma}
\label{lem:open-image}
Let
$
\boldh:\Complex^N\to\Complex^K
$
be holomorphic, with \(N\geq K\). If
$
\rank{ \boldJ_{\boldh}\left (\Parameters_0 \right )}=K
$ for some $\Parameters_0$, 
then \(\boldh\left (\Complex^N\right )\) contains a nonempty Euclidean-open subset of
\(\Complex^K\).
\end{lemma}

\begin{proof}
Since \(\boldJ_{\boldh}(\Parameters_0)\) has rank \(K\), there is an \(K\times K\) minor of the Jacobian matrix which is nonzero at \(\Parameters_0\). After a permutation of the coordinates, assume that
$$
\det
\left[
\frac{\partial h_i}{\partial \phi_j}(\Parameters_0)
\right]_{i,j=1}^K
\neq0 \, .
$$

Fix the remaining coordinates at their values at \(\Parameters_0\), and define the map
$
\boldg:\Complex^K\to\Complex^K
$
by
$
\boldg(\phi_1,\ldots,\phi_K)
=
\boldh(\phi_1,\ldots,\phi_K,\phi_{0,K+1},\ldots,\phi_{0,N}).
$
Then
$
\det \boldJ_{\boldg}(\Parameters_0)\neq0.
$
By the holomorphic inverse function theorem, there exist
neighborhoods \(\Vcal\) of
$
(\phi_{0,1},\ldots,\phi_{0,K})
$
and \(\Wcal\) of
$
\boldy_0=\boldh(\Parameters_0)
$
such that
$
\boldg:\Vcal \to \Wcal
$
is biholomorphic. 
Hence
$
\Wcal=\boldg(\Vcal)\subseteq \boldh \left (\Complex^N \right ).
$
Thus \(\boldh \left (\Complex^N \right )\) contains the nonempty Euclidean-open set \(\Wcal\).
\end{proof}

\subsection{Step 2: An Euclidean open image is Zariski dense}
\begin{lemma}
\label{lem:zariski-dense}
Let \(\Ucal\subseteq\Complex^K\) contain a nonempty Euclidean-open set $\Wcal$. Then $\overline{\Ucal}^{\,\mathrm{Zar}}=\Complex^K.$
\end{lemma}

\begin{proof}
Suppose that a polynomial $p$
vanishes on \(\Ucal\). Since \(\Ucal\) contains a nonempty Euclidean-open set \(\Wcal\), we have
$p(\Parameters)=0, \, \forall \Parameters\in \Wcal$. 
A polynomial in several complex variables which vanishes on a
nonempty Euclidean-open set must vanish identically. Hence $p\equiv0$.
Therefore, there is no nonzero polynomial whose zero set contains \(\Ucal\). By definition of the Zariski closure,
$
\overline{\Ucal}^{\,\mathrm{Zar}}=\Complex^K.
$
\end{proof}

\subsection{Step 3: Chevalley's theorem and generic surjectivity}


\begin{theorem}[Chevalley's theorem]
Let $\boldh:\Acal\to \Ucal $ be a polynomial or, more generally, a regular map of algebraic varieties. Then the image of every constructible subset of \(\Acal\) is a constructible subset of \(\Ucal\).
In particular, the image \(\boldh(\Acal)\) is constructible \cite{Hartshorne:B77}.
\end{theorem}

\begin{lemma}
\label{lem:constructible-dense}
Let \(\Ccal \subseteq\Complex^K\) be constructible and suppose that $\overline{\Ccal}^{\,\mathrm{Zar}}=\Complex^K$ (i.e., it is Zariski dense). 
Then \(\Ccal\) contains a nonempty Zariski-open subset of \(\Complex^K\).
\end{lemma}

\begin{proof}
Since \(\Ccal\) is constructible, it is a finite union of locally closed
sets:
\begin{equation}
\Ccal=\bigcup_{k=1}^K \left (\Ucal_k\cap \Zcal_k \right)
\end{equation}
with \(\Ucal_k\) Zariski open and \(\Zcal_k\) Zariski closed.

Since \(\Ccal\) is Zariski dense, it is
\begin{equation}
\Complex^K= \overline{\Ccal}^{\,\mathrm{Zar}}
\subseteq \bigcup_{k=1}^K \Zcal_k.
\end{equation}

The affine space \(\Complex^K\) is irreducible, so it cannot be a
finite union of proper Zariski-closed subsets. Hence $\Zcal_{k_0}=\Complex^K$ 
for some \(k_0\). Consequently, $\Ucal_{k_0}\cap \Zcal_{k_0}=\Ucal_{k_0} \subseteq \Ccal.$ 
Thus \(\Ccal\) contains a nonempty Zariski-open subset.
\end{proof}

An interesting property is given by the following lemma 
 \begin{lemma} \label{lem:dense} Let \(\Ucal \subset\Complex ^K\) be a nonempty Zariski-open set. Then \(\Ucal\) is dense in the Euclidean topology. Moreover, \(\Complex^K\setminus \Ucal\) has Lebesgue measure zero, i.e., $\operatorname{Leb}_{2K} \left ( \Complex^K\setminus \Ucal\right )=0$.
 \end{lemma} \begin{proof} See \cite{GortzWedhorn:B10,Kra:B01}. \end{proof}

We now have all the necessary ingredients to prove Theorem \ref{thm:Sufficient}.

\begin{proof}[Proof of Theorem \ref{thm:Sufficient}]
\label{thm:ae-surjectivity}
Let
\begin{equation}
\boldh:\Complex^N\longrightarrow\Complex^K,
\qquad N\geq K,
\end{equation}
be a polynomial map, with $K=MS$, and that the assumptions of Theorem \ref{thm:Sufficient} are satisfied. 

By Lemma~\ref{lem:open-image}, the image
\(\boldh\left (\Complex^N \right )\) contains a nonempty Euclidean-open set. By Lemma~\ref{lem:zariski-dense}, this implies
\begin{equation}
\overline{\boldh\left (\Complex^N \right )}^{\,\mathrm{Zar}}=\Complex^K \, .
\end{equation}
Hence \(\boldh\) is dominant. Since \(\boldh\) is polynomial, 
by Chevalley's theorem $\boldh\left (\Complex^N \right )$
is constructible. Since it is also Zariski dense, Lemma~\ref{lem:constructible-dense} implies that there exists a nonempty Zariski-open subset  $\Ucal\subseteq\Complex^K$
such that $\Ucal\subseteq \boldh \left (\Complex^N \right)$. 
Therefore, by Lemma \ref{lem:dense}, it follows that $\Ucal=\boldh\left (\Complex^N \right )$ is dense and 
\begin{equation}
\label{eq:Leb}
\operatorname{Leb}_{2M}
\left(
\Complex^K\setminus \boldh \left (\Complex^N \right )
\right)=0 \, .
\end{equation}
In other words, \(\boldh\) is surjective almost everywhere in
\(\Complex^K\).
\end{proof}

\begin{remark}
Summarizing, the full-rank
condition at one point gives local openness of the image. This implies
Zariski density, but does not by itself give almost-everywhere
surjectivity for a general holomorphic map.
Polynomiality is used in the second part of the proof: it makes
\(\boldh\) an algebraic morphism, so that Chevalley's theorem implies that
the image is constructible. A constructible Zariski-dense subset of
\(\Complex^K\) necessarily contains a nonempty Zariski-open subset,
whose complement has Lebesgue measure zero. 
It is worth noticing that the same result holds also for $\boldh$ being a rational map provided that $\Complex^N$ is substituted with the actual domain of $\boldh$.
\end{remark}

\ifCLASSOPTIONcaptionsoff
\fi
\bibliographystyle{IEEEtran}

\bibliography{IEEEabrv,Biblio/BiblioDD,Biblio/MetaSurfaces,Biblio/EMInformationTheory,Biblio/IntelligentSurfaces,Biblio/MassiveMIMO,Biblio/MIMO,Biblio/THzComm,Biblio/EMTheory,Biblio/WINS-Books,Biblio/Vari}

\end{document}

%% file: acronyms.tex
\acrodef{ESPAR}{electrically steerable passive array radiator}

\acrodef{RA}{reconfigurable antenna}

\acrodef{ZF}{zero forcing}

\acrodef{SIM}{stacked intelligent metasurface}

\acrodef{ISAC}{integrated sensing and communication}

\acrodef{SE}{spectral efficiency}

\acrodef{DAC}{digital-to-analog converter}

\acrodef{DMA}{dynamic metasurface antenna}

\acrodef{SINR}{signal-to-interference noise ratio}

\acrodef{ESIT}{electromagnetic signal and information theory} 

\acrodef{ELAA}{extremely large antenna arrays} 

\acrodef{DSA}{dynamic scattering array}

\acrodef{ULA}{uniform linear array}

\acrodef{UCA}{uniform circolar array}

\acrodef{IIoT}{industrial Internet-of-things}

\acrodef{IT}{information theory}

\acrodef{SRE}{smart radio environment}

\acrodef{EMO}{electromagnetic object}

\acrodef{SVD}{singular value decomposition}

\acrodef{PSWF}{prolate spheroidal wave function}

\acrodef{CR}{channel response}

\acrodef{BS}{base station}

\acrodef{MS}{mobile station}

\acrodef{UE}{user equipment}

\acrodef{MIMO}{multiple-input multiple-output}

\acrodef{MISO}{multiple-input single-output}

\acrodef{RIS}{reconfigurable intelligent surface}

\acrodef{IRS}{intelligent reconfigurable surface}

\acrodef{LIS}{large intelligent surface}

\acrodef{MIS}{medium intelligent surface}

\acrodef{SIS}{small intelligent surface}

\acrodef{DoF}{degrees-of-freedom}

\acrodef{AF}{amplify \& forward}

\acrodef{DF}{detect \& forward}

\acrodef{JF}{just forward}

\acrodef{CSI}{channel state information}

\acrodef{RV}{random variable}

\acrodef{i.i.d.}{independent, identically distributed}

\acrodef{PSD}{power spectral density}

\acrodef{PDF}{probability distribution function}

\acrodef{CDF}{cumulative distribution function}

\acrodef{ch.f.}{characteristic function}

\acrodef{AWGN}{additive white Gaussian noise}

\acrodef{RSSI}{received signal strength indicator}

\acrodef{SNR}{signal-to-noise ratio}

\acrodef{LRT}{likelihood ratio test}

\acrodef{GLRT}{generalized likelihood ratio test}

\acrodef{GML}{generalized maximum likelihood}

\acrodef{LOS}{line-of-sight}

\acrodef{NLOS}{non-line-of-sight}

\acrodef{GDOP}{geometric dilution of precision}

\acrodef{GPS}{Global Positioning System}

\acrodef{FIM}{Fisher information matrix}

\acrodef{PEB}{position error bound}

\acrodef{WSN}{Wireless Sensor Network}

\acrodef{MAC}{medium access control}

\acrodef{RSS}{received signal strength}

\acrodef{RTT}{round-trip time}

\acrodef{MIMO}{multiple-input multiple-output}

\acrodef{MF}{matched filter}

\acrodef{ED}{energy detector}

\acrodef{ML}{maximum likelihood}

\acrodef{NL}{nonlinear}

\acrodef{MSE}{mean square error}

\acrodef{RMSE}{root mean square error}

\acrodef{ppm}{part-per-million}

\acrodef{PRP}{pulse repetition period}

\acrodef{ACK}{acknowledge}

\acrodef{UWB}{ultrawide bandwidth}

\acrodef{TNR}{threshold-to-noise ratio}

\acrodef{LOS}{line-of-sight}

\acrodef{LS}{least squares}

\acrodef{IR-UWB}{impulse radio UWB}

\acrodef{FCC}{Federal Communications Commission}

\acrodef{TH}{time-hopping}

\acrodef{PPM}{pulse position modulation}

\acrodef{PAM}{pulse amplitude modulation}

\acrodef{MUI}{multi-user interference}

\acrodef{PDP}{power delay profile}

\acrodef{PPP}{Poisson point process}

\acrodef{DS}{delay spread}

\acrodef{CED}{channel excess delay}

\acrodef{BPZF}{band-pass zonal filter}

\acrodef{SIR}{signal-to-interference ratio}

\acrodef{RFID}{radio frequency identification}

\acrodef{WPAN}{wireless personal area networks}

\acrodef{WWLB}{Weiss-Weinstein lower bound}

\acrodef{DP}{direct path}

\acrodef{MF}{matched filter}

\acrodef{MMSE}{minimum-mean-square-error}

\acrodef{SBS}{serial backward search}

\acrodef{NBI}{narrowband interference}

\acrodef{WBI}{wideband interference}

\acrodef{INR}{interference-to-noise ratio}

\acrodef{CIR}{channel impulse response}

\acrodef{ISI}{inter-symbol interference}

\acrodef{CPR}{channel pulse response}

\acrodef{LRT}{likelihood ratio test}

\acrodef{MUI}{multi-user interference}

\acrodef{EM}{electromagnetic}

\acrodef{CW}{continuous wave}

\acrodef{RF}{radiofrequency}

\acrodef{FCC}{Federal Communications Commission}

\acrodef{EIRP}{effective radiated isotropic power}

\acrodef{RCS}{radar cross section}

\acrodef{BAV}{balanced antipodal Vivaldi}

\acrodef{PRake}{partial Rake}

\acrodef{RTLS}{real time locating system}

\acrodef{CRB}{Cram\'{e}r-Rao bound}

\acrodef{ZZB}{Ziv-Zakai bound}

\acrodef{TOA}{time-of-arrival}

\acrodef{TOF}{time-of-flight}

\acrodef{WSN}{wireless sensor network}

\acrodef{MAC}{medium access control}

\acrodef{RSS}{received signal strength}

\acrodef{TDOA}{time difference-of-arrival}

\acrodef{RF}{radiofrequency}

\acrodef{RTT}{round-trip time}

\acrodef{AOA}{angle-of-arrival}

\acrodef{MF}{matched filter}

\acrodef{ED}{energy detector}

\acrodef{ML}{maximum likelihood}

\acrodef{MUR}{Multistatic radar}

\acrodef{HDSA}{high-definition situation-aware}

\acrodef{RRC}{root raised cosine}

\acrodef{OFDM}{orthogonal frequency division multiplexing}

\acrodef{IF}{intermediate frequency}

\acrodef{PHY}{physical layer}

\acrodef{S-V}{Saleh-Valenzuela}

\acrodef{UHF}{ultra-high frequency}

\acrodef{PR}{pseudo-random}

\acrodef{SoC}{System on Chip}

\acrodef{SoP}{System on Package}

\acrodef{SPMF}{Single-Path Matched Filter}

\acrodef{IMF}{Ideal Matched Filter}

\acrodef{SCR}{signal-to-clutter ratio}

\acrodef{BEP}{bit error probability}

\acrodef{BER}{bit error rate}

\acrodef{WSR}{wireless sensor radar}

\acrodef{HPBW}{half power beam width}

\acrodef{LEO}{localization error outage}

\acrodef{WSS}{wide-sense stationary}

\acrodef{TR}{time-reversal}

\acrodef{WSSUS}{WSS with uncorrelated scattering}

\acrodef{GP}{Gaussian process}

\acrodef{IMU}{inertial measurement unit}


%% file: def_EM_SP.tex
\newtheorem{theorem}{Theorem}[section]
\newtheorem{lemma}[theorem]{Lemma}
\newtheorem{proposition}[theorem]{Proposition}
\newtheorem{corollary}[theorem]{Corollary}
\newtheorem{definition}[theorem]{Definition}

\newcommand{\rank}[1]{{\rm rank} \left \{ #1 \right \}}

\newcommand{\diag}[1]{{\rm diag} \left \{ #1 \right \}}

\newcommand{\boldA} {{\bf{A}}}
\newcommand{\boldg} {{\bf{g}}}

\newcommand{\bolde} {{\bf{e}}}

\newcommand{\boldW} {{\bf{W}}}

\newcommand{\boldJ} {{\bf{J}}}
\newcommand{\boldH} {{\bf{H}}}

\newcommand{\boldB} {{\bf{B}}}

\newcommand{\boldI} {{\bf{I}}}

\newcommand{\boldy} {{\bf{y}}}
\newcommand{\boldh} {{\bf{h}}}

\newcommand{\Ucal} {\mathcal{U}}
\newcommand{\Acal} {\mathcal{A}}
\newcommand{\Ccal} {\mathcal{C}}
\newcommand{\Wcal} {\mathcal{W}}
\newcommand{\Vcal} {\mathcal{V}}
\newcommand{\Zcal} {\mathcal{Z}}

\newcommand{\Complex} {\mathbb{C}}

\newcommand{\Ptx} {P_{\text{T}}}

\newcommand{\Parameters} {\boldsymbol{\varphi}}

\newcommand{\Parametersl} {\boldsymbol{\varphi}_{\ell}}
\newcommand{\Parameterlm} {\varphi_{\ell,m}}